\documentclass[12pt]{article}

\usepackage{amsmath}
\usepackage{amssymb}
\usepackage{tikz}
\usetikzlibrary{arrows,calc,shapes,decorations.pathreplacing,patterns}
\usepackage{pgfplots}
\usepackage{amssymb}
\newtheorem{theorem}{Theorem}[section]

\newtheorem{corollary}[theorem]{Corollary}

\def\endpf{{\ \hfill\hbox{\vrule width1.0ex height1.0ex}\parfillskip 0pt \\
}}
\newenvironment{proof}{\noindent{\bf Proof:}}{\endpf}

\def\E{\mathbb{E}}

\def\l{\lambda}

\def\eps{\epsilon}

\def\m{\mu}
\def\r{\rho}

\def\t{\theta}

\def\P{\mathbb{P}}

\newcommand{\qed}{\nobreak \ifvmode \relax \else
	\ifdim\lastskip<1.5em \hskip-\lastskip
	\hskip1.5em plus0em minus0.5em \fi \nobreak
	\vrule height0.75em width0.5em depth0.25em\fi}

\begin{document}

		\title {An optimal mechanism charging for  priority in a queue}
		\author{Moshe Haviv\footnote{Department of Statistics and Data Science, and  the Federmann Center for the Study of Rationality, the Hebrew University of Jerusalem} \ and  Eyal Winter\footnote{ The Management School, the University of Lancaster and the Department of Economics and  the Federmann Center for the Study of Rationality, The Hebrew University of Jerusalem}
		}
		\maketitle
		
		\begin{abstract}

We derive a revenue-maximizing scheme that charges customers  who are homogeneous with respect to their waiting cost parameter for a random fee in order to become premium customers. This scheme incentivizes all customers to purchase priority, each at his/her drawn price. 
We also design a revenue-maximizing scheme for the case where customers are
heterogeneous with respect to their waiting cost parameter.
Now lower cost parameter customers are encouraged to join the premium class at a low price: Given that, those with high cost parameter  would be willing to pay even more for this privilege.

	\end{abstract}
		
	%priority  increases with the measure of priority customers Also, in equilibrium,
	%all (and hence none) belong to the premium class.
	%We show that when customers are homogeneous with respect to their the waiting %cost parameter, each one of them is charged the price that would make him/her %indifferent between buying and not buying, under the belief that all customers %appearing before him/her in some random order have already purchased priority. %In fact, this is the maximin equilibrium for the profit maximizer.

\section{Introduction}
There are numerous examples in which service is granted to an order based not on the standard first-come first-served (FCFS) regime, but rather  on  priority levels held by the customers. For example, in an emergency room patients who need life-saving care are treated before those who complain of a minor fever. Another example is early boarding granted to premium customers 
by airlines. By paying a priority price these passengers can get better seats (if prior seat allocation is not done) and have guaranteed space for carry-on luggage. 

These two examples indicate two reasons for policing a priority queue regime. The first shows that a social objective can be improved by discriminating in favor of some customers based on individual costs due to waiting. In the queueing literature this corresponds to the well-known $C\m$-rule: service is prioritized based on an index each customer possesses. The value of the index is the product of the cost per unit of time in the queue (usually denoted by $C$) and the inverse of the mean service time (denoted by $\m$). See, e.g.,~\cite{Hav13}, pp. 72--73. Here, as in the emergency room, those who suffer more from waiting get priority over those who suffer less. Similarly, since it is socially better that a long task  waits for a short one, rather than the other way around, short tasks have priority over long ones.

The second example shows that firms or monopolies can generate profit by offering priority based on some charge. Customers who compete among themselves for priority might be willing to pay more for it. They might be more willing to do so if they suffer more from waiting and wish to overtake others, or if they suffer from the idea of being overtaken by others, and wish to remove the threat. Note that no extra effort or cost is incurred by the profit-maximizer who administrates the priority system but nevertheless generates more profit to itself. Also note that if all decide to purchase priority,  priority will in fact be granted to none (and payment will not be reimbursed). 

The above examples illustrate  the case where two priority levels exist; however, one can administrate any finite or continuous number of levels,  each of which comes with its own price. 

Another question is the implementation related to the administration of a priority scheme. For example, in order to implement the $C\m$-rule, a central planner needs to know the cost parameter of each of the arriving customers. Yet, asking customers for their cost parameter $C$ comes with the risk that customers will cheat
and represent themselves as having a higher value for $C$
than their actual one.
It is shown in~\cite{MW} and~\cite{AM} that there  exists a price menu for priority levels where purchasing the true priority level across customers is a Nash equilibrium profile. We omit here the exact details of the models described there.

We also like to note that  by administrating a priority
scheme one can  regulate the arrival rate. Usually customers over-congest a queueing system. It is shown in~\cite{Has95} that if customers who decide to join a queue compete for priority levels, so that as the more they pay, the higher  their priority, they end up joining the system at the socially optimal rate. The same phenomenon exists in the model presented in~\cite{MW}. We also like to mention~\cite{HO} where is it shown how a random priority scheme regulates the system in the sense that it results in a socially optimal joining rate. For the observable version of this model, 
a menu of prices, one for each priority level, is designed in~\cite{Al}.
This results in customers  joining the queue only when it is socially optimal to do so.

Finally,   the priority scheme described above (where  customers possess a priority parameter and the one who holds the highest value is the first to be granted service) is not the only one which exists. For example, \cite{HV} and~\cite{HasHav08} deal with a relative priority model in which the probability that one enters service next is proportional to one's priority parameter. Another model is based on accumulating priority. Here customers increase their priority level with the time they are in the queue (as in FCFS) but the rate of increase is individual. Thus, a higher rate reflects a higher priority level. In particular, this model grants priority based both on some priority parameter and on seniority. For more details
see~ \cite{HR} and~\cite{AHOZ} for the cases where the menu of choices being continuous or discrete, respectively.

This paper considers the profit-maximizer point of view where, as in the airline example, there are two priority levels in a queue, say premium and ordinary classes. Here the objective is to design a price mechanism under which the highest possible  revenue is generated. 
The primitive assumption is that once a priority mechanism is introduced, the customers decide, whether or not to pay the price offered to them for belonging to the premium class. 
Of course, they decide while minding only their selfish interest. 
Clearly, for any price mechanism, customers are engaged in a noncooperative game among themselves and we look for the resulting Nash equilibrium, which in turn determines the monopoly's revenue. As we discuss below, the existence of multiple equilibria is common and  indeed this is the case where a uniform charge is offered. See~\cite{HasHav}, pp. 83--85 for further details. In the case of of multiple equilibria, we judge a price mechanism based on the worst equilibrium in terms of the revenue it generates to the operator (a maxmin  criterion).
 
 We deal with two cases. The first, in Section~3,  is when customers are homogeneous with respect to their
cost per unit of time parameter $C$. The second is when they are not homogeneous and hence $C$ can be looked at as a random variable. In both cases we design the maxmin revenue mechanism. In the homogeneous  case, the optimal mechanism charges a random price whose distribution function is derived in Section~3. In the heterogeneous case the optimal charge is cost-parameter dependent. As expected, the higher the cost-parameter, the higher  the charge is. Details are given in Section~4. The queueing model and the required preliminaries from queueing theory are given in Section~2. Section~5 concludes.

\section{The model}
We consider the memoryless single-server $M/M/1$ queue. 
Specifically, there exists a Poisson arrival process with
a rate denoted by $\l$. Service times follow an exponential
distribution with rate denoted by $\m$.
Denote $\l/\m$ by $\r$ and assume for stability (positive recurrence) that $\r<1$. 
Customers belong to two priority groups, where premium customers have preemptive priority over ordinary customers. Within a class, we assume the first-come first-served (FCFS) regime. Suppose that a fraction $q$ of the customers
belong to the priority class, while the rest are ordinary customers. It is well-known, see e.g., \cite{Hav13}, p.~75, that the mean waiting time (service inclusive)
at the former group
is $$W_1(q)=\frac{1}{\m(1- q \r)},$$
while that of the latter group  is
$$W_2(q)=\frac{1}{\m(1-\r)(1- q \r)}.$$
It is then easy to see that
\begin{equation} \label{f(p)}
f(q)=W_2(q)-W_1(q)=\frac{\r}{\m(1-\r)(1-q\r)}.
\end{equation}
As pointed out in~\cite{HasHav}, p.~84, this function is monotone increasing with $q$, $0 \leq q \leq 1$. This implies that the larger  the group of premium customers, the more valuable it is for an individual customer to belong to this group. This is the follow the crowd  (FTC) phenomenon. See~\cite{HasHav}, p.~6, for more details.

Seen  from the individual customer's point of view, and measured in units of time, the value of priority is bounded between $f(0)$ and $f(1)$:
$$f(0)=\frac{\r}{\m(1-\r)} \leq f(q) \leq \frac{\r}{\m(1-\r)^2}=f(1)$$
and, by definition, it is a function of the fraction of customers who belong to the priority class.
On the one hand,
$$f(0)=\frac{1}{\m(1-\r)}-\frac{1}{\m}.$$
This is the gain for one who solely gets top priority as opposed to being an average customer, for example, one who waits in a FCFS queue. On the other hand, 
\begin{equation} \label{f1}
f(1)=\frac{1}{\m(1-\r)^2}-\frac{1}{\m(1-\r)}.
\end{equation}
This is the gain from being an average customer as opposed of being a standby customer (i.e., getting  service only when the server would otherwise be idle). Put differently, one can get this gain if one switches from  the  ordinary class to the  premium class, when all are premium customers. This switch is when priority is most valuable.
 See~\cite{Hav13}, p~.64 and pp.~76--77, for more on standby customers.

 In~\cite{HasHav}, p~.83, a decision model, in fact a non-cooperative game, where customers are charged for belonging to the premium class (where the default is to become an ordinary customer) is introduced. For a comprehensive review on the literature on queueing games see~\cite{Has16}. Denote this common charge by $\t$.
 Due to the monotonicity of $f(q)$, it is clear that the solution to the dilemma of whether or not to purchase priority is simple
in the case where $\t \leq f(0)$: priority should be purchased. This is a dominant strategy: regardless of what the others do,
one is better off paying $\t$ and joining the premium class. Of course, the end result is that all follow this strategy and the resulting queue regime is FCFS. An analogous situation happens in the case where $\t \geq f(1)$: nobody purchases priority 
as this is the dominant strategy. Here too the resulting queue regime is FCFS.

The situation is more involved when $f(0) < \t < f(1)$. Now no dominant strategy exists. In particular, the best response for an individual depends on the fraction $q$. Denote by
$q_e$ the unique value for $q$ that obeys
$f(q_e)=\t$. Then, if $q < q_e$ the best response is not to purchase priority, while if $q>q_e$, it is  to purchase. In the case where $q=q_e$, either purchasing or not, is one's best response. In particular, one is  indifferent between these two options. In fact, any mixing between the two is also one's best response. 

Once a dominant strategy does not exist, one looks for a symmetric Nash equilibrium,  namely, a strategy that if used by all but one player,  is   one's best response.
As shown in~\cite{HasHav}, p.~83, there are three equilibria when $\t \in (f(0),f(1))$: when all purchase, when none do, and when all purchase with probability $q_e$ (namely, a fraction of $q_e$ of the customers become premium customers). Which one of these three equilibria will emerge is not clear. In terms of the stability of the equilibrium, it is possible to see that the equilibrium that is based on mixing with probabilities $q_e$ and $1-q_e$ is unstable: a small deviation in one direction tilts the balance in the same direction. For example, if $q_e+\epsilon$ for $\epsilon>0$ is adopted by all but one player, the unique best response is to purchase too.  In this respect,
the two pure equilibria are stable while is the mixed one is not.  For more on the issue of the stability of equilibria, in particular for the definition of evolutionarily stable strategies (ESSs), see, e.g.,~\cite{HasHav}, p.~15. There is no answer to  question of which equilibrium will be selected (even among the stable ones), so, when judging  a price mechanism, we look at the revenue it generates under the worst possible equilibria. 
  	In particular, we call a price mechanism revenue-maximizer, or optimal,
 		 if the memchanism generates a certain revenue for the monopoly in a unique equilibrium among  customers and there exists no other price mechanism that  yields a higher revenue via in a unique equilibrium.
 	 	Under this definition, charging more than $f(0)$ backfires as  under the worst equilibrium, nobody pays. The optimal flat price is hence $f(0)$. As we will show in the next section, there exist random price mechanisms that generate more income.

\section{Profit maximization: Homogeneous customers }
The main purpose of this article is to introduce a price mechanism that leads, under the resulting unique Nash equilibrium in the game among customers, to a  gain per customer that is greater than $f(0)$ (but still less than $f(1)$). Moreover, the suggested mechanism is optimal. 
The mechanism is based on a random entry fee. Specifically, each arrival will
be asked to draw a random number from a to-be-determined distribution whose support is the interval $[f(0),f(1)]$. Accepting this fee will be the unique equilibrium among customers and the resulting queue regime will again be FCFS. 
Note that in order to ease the exposition we assume that whenever customers are indifferent between purchasing priority or not, they elect for the former option.

\begin{theorem} \label{payment}
	Suppose that the charge for priority is a random
	variable with CDF $F(p)$, where
	\begin{equation} \label{F(p)}
	F(p)=\left\{\begin{array}{cc}
	0 & p \leq f(0), \\
	\frac{1}{\r}-\frac{1}{\m(1-\r)}\frac{1}{p} & f(0) \leq p \leq f(1), \\
	1 & p \geq f(1).
	\end{array}	\right.
	\end{equation}
If, for some  given price $p$, $f(0) \leq p \leq f(1)$, all  who are asked to pay a price that is less than  $p$ do so, then it is uniquely best for one who is asked to pay $p$ to do so  as well, regardless of what those who are asked to pay a price greater then $p$ actually do. Note that a tie exists only in the case when all from the latter group do not pay, and remain ordinary customers.   
\end{theorem}

\begin{proof}
Consider a random price mechanism with a cumulative distribution function (CDF) $F(p)$. For a given value for $p$, suppose all those who are asked to pay a price  less than $p$ do so. Then, under the lowest possible value for priority for one who is asked to pay $p$, namely, when all who are asked to pay a price greater than $p$ do not do so, the value of priority, based on~(\ref{f(p)}), equals
$$\frac{\r}{\m(1-\r)(1-F(p)\r)}, \ \ f(0) 
\leq p \leq f(1).$$
Thus, if one is asked to pay $p$, where
\begin{equation}\label{fixpoint}
p=\frac{\r}{\m(1-\r)(1-F(p)\r)},
\end{equation}  
 	then one should still pay. (In practice, one will be asked to pay a little bit less but for the sake of ease of exposition we ignore this issue.) Indeed, one is indifferent between paying or not if and only if all customers who are asked to pay a price greater than $p$ do not do so.  
Solving~(\ref{fixpoint}) for $F(p)$ as a function of $p$
shows that~(\ref{F(p)}) holds.  In particular, $F(f(0))=0$ and $F(f(1))=1.$	
\end{proof}

\noindent {\bf Remark.}
	The mean payment based on the scheme stated in Theorem~\ref{payment} equals
	\begin{equation} \label{MeanP}
	-\frac{\log(1-\r)}{\m(1-\r)}.
	\end{equation}
Of course,
$$f(0)=\frac{\r}{\m(1-\r)}<-\frac{\log(1-\r)}{\m(1-\r)}<
\frac{\r}{\m(1-\r)^2}=f(1).$$

\noindent {\bf Remark.}
A way to implement the above scheme is to announce that the customer who will be in the $F(p)$ percentile to purchase priority will be offered to pay  $p$ for 
it, $f(0) \leq p \leq f(1)$. Customers can be looked at as being in competition over who will be the first to pay. The earlier they do so relative to others, the less they will have to pay. Another interpretation is as follows. Suppose that customers are somehow ordered. The first is offered priority at the price of $f(0)$ (which he/she accepts and pays), the second is offered priority at a somewhat higher price (and is willing to pay more than the first customer since s/he wishes not to be overtaken by the first), the third is offered priority at an even higher price, which s/he duly pays (in order not to be overtaken by the first two is), etc. 

\noindent{\bf Remark.}	
It is possible to see that if instead of $F(p)$ one  implements another CDF $G(p)$ with $G(p) \geq F(p)$ for 
$p$, the same customers' behavior will be induced. Yet, a reduction in the payment per customer will occur. This can easily be argued from the fact that if a nonnegative random variable comes with a CDF $G(p)$, then its mean value equals $\int_{p=0}^{\infty}(1-G(p))\,dp$. See, e.g.,~\cite{Hav13}, p~.3. The following theorem says even more: the scheme suggested in Theorem~\ref{payment} is revenue maximizer.

\begin{theorem} \label{optimal}
	Among all random price mechanisms that lead to the ``all pay" profile being the unique equilibrium, the one suggested in Theorem~\ref{payment} is optimal. Moreover, if a random mechanism comes with a CDF $G(p)$ in which ``all pay" is the unique equilibrium, then $G(p) \geq F(p)$ for all $p \in R$.
\end{theorem}  

\begin{proof} Aiming for a contradiction, let $G(p)$ be another mechanism that comes with ``all pay"  as an equilibrium  profile whose  profit is higher than it is under $F(p)$. Let $p' =\inf_p \{ G(p) < F(p) \}$. Loosely speaking, $p'$ is the ``first"
	 charge where the probability of paying this value or below it is smaller under $G(\cdot)$  than it is under $F(\cdot)$. This is the first point where, at least locally, this mechanism becomes more expensive. Note that our assumption on $G(\cdot)$ implies the existence of $p'$. (Note that the option where $p'=0$ is not ruled out.) This implies that those charged  $p'$
	 are asked to pay more than needed under the assumption that  only a  fraction of $F(p')$ will pay. They, as well as those who are charge more, may now consider not paying. In such a case, it is possible to see that under this mechanism, the profile where all those who are charged up to $p'$ pay, while the rest do not, is a Nash equilibrium outcome, making it an additional equilibrium to the assumed ``all pay" one. In particular, there is no unique equilibrium under this charging scheme. 	 
\end{proof}

\subsection{The discrete approach}	
A possible mechanism is to discretize the random variable whose CDF is $F(p)$, in such a way that each discrete value gets all the probability mass of those values that are greater than it but are still less than the next discrete value. For example, for some choice of an integer~$n$, denote $1/n$ by $\eps$. Then, define
the following $n+1$ values for $p$: $p_i=F^{-1}(i \eps )$,
$0 \leq i \leq n$. Note that as $F^{-1}(y)=\r /\m(1-\r)(1-\r y)$, we get in fact that $p_i=  \r /\m(1-\r)(1-\r i \eps)$. In particular, $p_0=f(0)$ while $p_n=f(1)$. By this construction and by denoting by $P$ the random charge based on the original scheme, it is possible to see that $\P(p_i \leq P \leq p_{i+1})=F(p_{i+1})-F(p_i)=(i+1)\eps-i\eps=\eps$, $0 \leq i \leq n-1$.
Define now a uniform discrete random variable $X$ whose support is composed of the $n$ values $p_i$, with identical probabilities, i.e., 
$\P(X=p_i)=\eps$,  $0 \leq i \leq n-1$. Clearly, $P$ is stochastically dominated by $X$. 

By considering a discrete random charge for priority based on $X$, it is somewhat easier to see the rationale leading to all paying the prescribed price as its unique equilibrium. The argument in fact goes by induction. Specifically, those who are asked to pay $p_0=f(0)$ will certainly do so: they all are willing to pay that regardless of what others do. This fact is realized by those who are asked to pay~$p_1$. Knowing that a fraction of $\eps$ have already paid, they are willing to pay~$p_1$, regardless of what  those who are asked to pay more do. Hence, they  pay what they are asked to. This argument goes on until all pay. It is based on the idea that for any $\eps>0$ paying is always the unique best response for any realization for $X$, given all pay their realized~$X$. Note that
\begin{equation} \label{ExpX}
\E(X)=\eps \frac{\r}{\m(1-\r)}\sum_{i=0}^{n-1}
\frac{1}{1-\r i \eps}.
\end{equation}
Recalling that $\eps=1/n$ and taking the limit when $n$ goes to infinity in~(\ref{ExpX}), we get~(\ref{MeanP}), as expected.

\noindent {\bf Remark.} It is possible to argue for the resulting Nash equilibrium by the process of elimination of weakly dominated strategies. Specifically, as not purchasing priority is weakly dominated by purchasing for those who are offered the lowest price, the former strategy can be eliminated. Given that, the same is hence the case for one who is offered the second lowest price and one's not purchasing strategy can now be eliminated. This goes on until the one who is offered the highest price.

\subsection{When homogeneous customers select their priority level}
A related but different decision model is suggested in~\cite{GH}.
See also~\cite{HasHav}, p. 103. Now it is the customers  who decide (implicitly) on their priority level and
have to pay accordingly: the more they pay, the higher is their
priority level. Ties are broken randomly. The options are  all the nonnegative real numbers.
Modifying the analysis given in~\cite{GH} for the non-preemptive case, it can shown that  there exists a the unique equilibrium payment profile which comes with a random payment whose CDF, denoted by $B(y)$ equals
\begin{equation} \label{eqCDF}
B(y)=1-\frac{1}{\r}+\frac{1}{\r}(\frac{1}{(1-\r)^2}-\mu y)^{-\frac{1}{2}}, \ \ 
0 \leq y \leq \frac{1}{\m (1-\r)^2} -\frac{1}{\m}.
\end{equation}
Note that here customers are ex-ante identical but ex-post, due to the use of a mixed strategy, they are not.

Note that maximal possible payment equals the value for becoming a `sole child' who gets top priority for one who would otherwise be a standby customer. In other words, it is the value of moving from one  extreme situation (the worse) to the other (the best).
Indeed, being a standby customer is  the fate of the one who pays nothing when all pay something. 

The mean payment per customer is then
\begin{equation} \label{meaneq}
\int_{y=0}^{\infty}(1-B(y))\,dy=
\frac{1}{\r}\int_{y=0}^{\frac{1}{\m (1-\r)^2} -\frac{1}{\m}}
(1-(\frac{1}{(1-\r)^2}-y\mu)^{-\frac{1}{2}})\,dy=\frac{\r}{\m(1-\r)^2}.
\end{equation}

Recall that this value equals $f(1)$ (see~(\ref{f1})). It is possible to see that the ratio between the upper bound in~(\ref{eqCDF}) and~(\ref{meaneq})
equals (only) $2-\r$ which indicates that most of the mass of this distribution is at the upper end. Indeed, it is possible to see that $B'(y)$,
which is the corresponding density is monotone increasing, going up
from $(1-\r)^3/(2\m)$ at $0$ to $1/(2\m)$ at $\frac{1}{\m (1-\r)^2} -\frac{1}{\m}$.
One needs to compare~(\ref{MeanP}) with~(\ref{meaneq}) (the latter is larger) but note that the two models differ not only in who determines the price mechanism but, crucially, in
the fact that in the former model there exist only two priority levels while in the latter model there are continuously many. 
   	
\section{Profit maximization: Heterogeneous customers}
Assume now that customers are heterogeneous and they vary in terms of their waiting costs. Specifically, a type~$C$ customer suffers a cost of $C$ per unit of time in the system (service inclusive).  This heterogeneity is modeled
	by assuming that for each individual $C$ is  a nonnegative random variable with a density function $g(c)$ and a CDF $G(c)$, $G(c)=\int_{y=0}^cg(y)\,dy$.

A profit maximizer wishes  to charge as much as possible for priority. Suppose that it is possible to discriminate between customers and charge them for priority an amount  that depends on their $C$  value, denoted by $c$. It makes sense to charge a price that is monotone in $c$ as high-cost customers will be willing to pay more in order to save themselves the same amount of waiting. Moreover, a good profit-maximizing scheme  encourages, by way of a relatively small charge, customers with a lower value for $c$ to purchase priority (and to make it clear that this is what they  do) so as to make high-cost customers willing to pay even more for priority when they realize that the group of those who pay (and hence belong to the premium class) is increasing in size.

\begin{theorem}
	The price mechanism that asks a customer with waiting cost parameter $c$ to pay 
	\begin{equation} \label{payhetero}
	\frac{\r c}{\m(1-\r)(1-G(c)\r)}
	\end{equation}
	for priority leads to ``all pay" being the unique Nash equilibrium for customers. Moreover, this is optimal among all price mechanisms that lead to a unique equilibrium from the customers' side.
\end{theorem}

\begin{proof}
Assume that all customers with a cost--parameter $C$ that is smaller than $c$ belong to the premium class. Then, a customer whose cost--parameter equals~$c$ agrees to pay (at least)
what is prescribed in~(\ref{payhetero})
as this is the lowest possible value for priority for him/her (since all those with a higher value for $C$ do not pay for priority). See~(\ref{f(p)}).
A similar argument to the one given in 
the proof of Theorem~\ref{optimal} for the optimality of the mechanism holds here too and hence will not be repeated.
\end{proof} 

The following corollary now follows.

\begin{corollary}
 The optimal profit equals 
\begin{equation} \label{optimalhetero}
 \frac{\r}{\m(1-\r)}\int_{c=0}^{\infty}\frac{c}{1-G(c)\r}g(c)\,dc.
\end{equation}
It is bounded from below and from above by 
$$\frac{\r\E(C)}{\m(1-\r)} \ \ \mbox{and} \ \ 
\frac{\r\E(C)}{\m(1-\r)^2}
,$$
respectively.
\end{corollary}

\noindent {\bf Remark.}
This mechanism can be looked at as bribing customers with a low value for $C$ to join the premium class at a low charge, making those
with a high value for $C$ face a tougher situation in case they stay behind in the race for priority. This will make them willing to pay even more than what might at  first sight be considered  as a reasonable  price.

\subsection{When heterogeneous customers select their priority
	parameter}

A related but different decision model is suggested in~\cite{GH}. See also~\cite{HasHav}, p~.103. Here  it is the customers who decide  how much to pay for priority,
having their own cost parameter as their private information. Moreover, the more one pays, the higher one's priority level is (breaking ties randomly).  Assuming that the random variable $C$ comes with a continuous non-zero density along its support, it was shown there that in the unique equilibrium profile, the higher  one's cost parameter is, the higher one's payment (and hence one's priority level) is. Modifying the analysis given in~\cite{GH} for the non-preemptive case, we get that  if one's cost per unit of time is $c$, one's payment equals
$$\frac{2\r}{\m}\int_{y=0}^c\frac{1}{(1-(1-G(y)\r)^3}g(y)\,dy.$$ 	
The mean payment per customer is then of course
\begin{equation} \label{eqhetero}
\frac{2\r}{\m}\int_{c=0}^\infty\int_{y=0}^c\frac{1}{(1-(1-G(y)\r)^3}g(y)\,dy g(c)\,dc.
\end{equation}

One needs to compare~(\ref{optimalhetero})
with~(\ref{eqhetero}) as both assume that  the set of 
priority levels is continuous. The disclaimer we made on such comparisons
at the end of Section~3 holds here too.

\begin{center}
	{\bf Acknowledgment}
\end{center}
 	The first author was patially supported by an Israel Research Fund, Grant no.~511/15.

\end{document}